\documentclass[11pt]{article}

\usepackage[textwidth=15cm]{geometry}

\usepackage{amsmath}
\usepackage{amsfonts}
\newcommand{\qed}{\penalty 1000 \hfill\penalty 1000{\small $\square$}\par\medskip}
\newenvironment{proof}{{\it Proof:}}{\qed}

\usepackage{amsmath}

\newcommand{\im}[2]{#1\langle #2 \rangle}
\newcommand{\et}[2]{\eta_{#1,#2}}
\newcommand{\rh}[2]{\rho_{#1 \rightarrow #2}}

\newcommand{\ep}[1]{\epsilon_{#1}}

\newcommand{\tw}{\mbox{\rm tw}}
\newcommand{\cw}{\mbox{\rm cw}}
\newcommand{\mcw}{\mbox{\rm mcw}}

\newcommand{\rw}{\mbox{\rm rw}}
\newcommand{\boolw}{\mbox{\rm boolw}}

\newtheorem{theorem}{Theorem}
\newtheorem{corollary}[theorem]{Corollary}
\newtheorem{proposition}[theorem]{\bf Proposition}
\newtheorem{claim}{Claim}
\newtheorem {definition}{Definition}

\title{Multi-Clique-Width}

\author{Martin F\"urer%
\thanks{Research supported in part by NSF Grant CCF-1320814.} \\
	Department of Computer Science and Engineering \\
	Pennsylvania State University \\
	University Park, PA 16802,  USA \\
	furer@cse.psu.edu \\}
\date{November 2, 2015}

\begin{document}

\maketitle

\begin{abstract}
Multi-clique-width is obtained by a simple modification in the definition of clique-width. It has the advantage of providing a natural extension of tree-width. Unlike clique-width, it does not explode exponentially compared to tree-width. Efficient algorithms based on multi-clique-width are still possible for interesting tasks like computing the independent set polynomial or testing $c$-colorability. In particular, $c$-colorability can be tested in time linear in $n$ and singly exponential in $c$ and the width $k$ of a given multi-$k$-expression. For these tasks, the running time as a function of the multi-clique-width is the same as the running time of the fastest known algorithm as a function of the clique-width. This results in an exponential speed-up for some graphs, if the corresponding graph generating expressions are given. The reason is that the multi-clique-width is never bigger, but is exponentially smaller than the clique-width for many graphs. This gap shows up when the tree-width is basically equal to the multi-clique width as well as when the tree-width is not bounded by any function of the clique-width. 
\end{abstract}

\newpage

\section{Introduction}
Tree-width is the first and by far the most important width parameter. It is motivated by the fact that almost all interesting problems that are hard for general graphs allow  efficient algorithms when restricted to trees. Furthermore such algorithms are often quite trivial. The promise of the notion of tree-width is to extend such efficient algorithms to much larger classes of tree-like graphs. Graphs of bounded tree-width have one shortcoming though. They are all sparse.

Clique-width \cite{CourcelleER93} is the second most important width parameter. 
It has been defined by Courcelle and Olariu \cite{CourcelleO2000} based on previously used operations \cite{CourcelleER93}.
It is intended to make up for the main shortcoming of the class of graphs of bounded tree-width. The idea is that many graphs are not sparse, but are still constructed in a somewhat simple and uniform way. One would expect to find efficient algorithms for such graphs too. The most extreme example is the clique. It's hard to find a natural problem that is difficult on a clique. 

It turns out that graphs of bounded tree-width actually also have bounded clique-width \cite{CourcelleO2000,CorneilR2005}, and many efficient algorithms extend to the larger class. Indeed, every graph property expressible in $\mathcal MS_1$, the monadic second order logic with set quantifiers for vertices only, is decidable in linear time for graphs of bounded clique-width \cite{CourcelleMR2000}.

The problem is that the containment of the bounded tree-width graphs in bounded clique-width graphs is not obvious. Furthermore, the generalization from bounded tree-width to bounded clique-width does not come cheap. The width can blow up exponentially, with a potential for a significant loss of efficiency for many algorithms. 

This creates a cumbersome situation for the many problems that have efficient solutions in terms of the tree-width as well as in terms of the clique-width, assuming the corresponding decompositions are known. One would like to run the algorithm based on clique-width to cover a much larger class of graphs, but that would mean an exponential sacrifice in running time for some graphs with small tree-width. Arguably, there should be a notion of a width parameter that bridges this gap more graciously. 
The current author has searched for some time for such a parameter. Ideally, we would like to have a natural generalization of tree-width and clique-width, and there should be no exponential blow-up in the parameter value.

The second objective has been obtained
with the notion of fusion-width \cite{Furer2014}. It has been shown before that the fusion operation does not produce unbounded clique-width graphs from bounded clique-width graphs \cite{CourcelleM2002}. Indeed, there is a much tighter relationship. Graphs of tree-width $k$ have fusion-width at most $k+2$ \cite{Furer2014}, while in the worst case, they have clique-width exponential in $k$ \cite{CorneilR2005}.

This is a very desirable property of fusion width.
The drawback is that attaching a fusion operation is somewhat unnatural. It is an artificial push of the tree-width concept into a clique-width-like environment.  Here, we present a far more natural width parameter that achieves this goal in a more direct way. We call it \emph{multi-clique-width}, it is obtained by a simple modification in the notion of clique-width, namely by allowing every vertex to have multiple labels. It turns out that multiply labeled graphs have been used before \cite{CourcelleO2000} in a more auxiliary role, and a variant of the multi-clique-width has actually appeared in the literature under the name m-clique-width \cite{CourcelleT2010} in the context of preprocessing for shortest path routing computations.

In this paper, we propose the multi-clique-width as a serious contender of clique-width. This powerful parameter has some very desirable properties. Its definition is equally simple and natural as that of clique-width. The multi-clique-width is never bigger than the clique-width, but often exponentially smaller \cite{CourcelleT2010}. And most importantly, there is no explosion of the width when moving from tree-width to multi-clique-width. Furthermore, there are interesting algorithms where the dependence of the running time on the (potentially much smaller) multi-clique-width is about the same as the dependence on the clique width for a similar algorithm working with clique-width. Thus, multi-clique-width allows some tasks to be solved much more efficiently than previously known.

There are other important width parameter, the rank-width \cite{OumS2006} and the boolean-width \cite{Bui-XuanTV2011}, that share some significant properties with the multi-clique-width. It too is never bigger than the clique-width and can be exponentially smaller. So why do we want to investigate yet another similar parameter? We claim that the rank-width serves a very different purpose than multi-clique-width.

Rank-width, boolean-width, clique-width, and multi-clique-width are all equivalent in the sense that the exact same problems are solvable in polynomial time for bounded width. If one of these parameters is bounded (by a constant), then so are the others. Rank-width has been introduced with this equivalence in mind \cite{OumS2006}. Before, the graphs of bounded clique-width could not be identified computationally. Therefore, graphs of bounded clique-width could only be handled efficiently, when a corresponding $k$-expression has been given. Now the rank-width $\rw(G)$ can be approximated, and a $2^{3\cw(G) + 2} -1$-expression can be computed efficiently \cite{OumS2006}. Still, the clique-width is the easiest to run many application algorithms. The rank-width is often not used directly for this purpose. In theoretical investigations, the exponential bound on the clique-width has not often been viewed as a major concern, because the goal has been to handle bounded clique-width graphs in polynomial time, not to speed them up, even when the rank-width is much smaller than the clique-width.

For our use of multi-clique-width, the motivation is different. We conjecture multi-clique-width to be NP-hard. Nevertheless it might be that the multi-clique-width could be approximated in polynomial time. We don't yet know. If so, multi-clique-width could serve the role of rank-width, but there is no need for this duplication. What is important in this case, is that multi-clique-width can easily be used directly in the design of some efficient algorithms. As a result, it can provide exponential speed-ups to application algorithms. These are exponential speed-ups in the parameter, meaning that the class of graphs with bounded parameter value would not change, just the computations get much faster.

In comparison with boolean-width, when the corresponding decompositions are given, we notice that the known efficient algorithms for NP-complete problems are quadratic in $n$ for bounded width, while for $\mathcal MS_1$
expressible graph properties, we have linear time algorithms \cite{CourcelleMR2000} for bounded clique-width or multi-clique-width. Furthermore, we will illustrate that such algorithms can be quite simple and efficient as a function of the multi-clique-width.

\section{Definitions and Preliminaries}
We use the standard notions of tree decomposition, tree-width, $k$-expression, and clique-width.
\begin{definition}
 A \emph{tree decomposition} of a graph $G=(V,E)$ is a pair $(\{B_{i } \, : \, i \in I \}, T)$, 
 where $T= (I,F)$ is a tree and each node $i \in I$ has a subset $B_{i} \subseteq V$ of vertices (called the bag of $i$) associated to it with the following properties.
 \begin{enumerate}
\item 
$\bigcup_{i \in I} B_{i} = V$, i.e., each vertex belongs to at least one bag.
\item 
For all edges $e=\{p,q\} \in E$, there is at least one $i \in I$ with $\{p, q\} \subseteq B_{i}$,
i.e., each edge is represented by at least one bag.
\item
For every vertex $v \in V$, the set of indices $i$ of bags containing $v$ induces a subtree of $T$
(i.e., a connected subgraph).
\end{enumerate}
\end{definition}

\begin{definition}
The \emph{width of a tree decomposition} is 1 less than its largest bag size.
 The \emph{tree-width} $\tw(G)$ \cite{RobertsonS84} of a graph $G$ is the width of a minimal width tree decomposition of $G$. 
\end{definition}

It is NP-complete to decide whether the tree-width of a graph is at most $k$ (if $k$ is part of the input) \cite{ArnborgCP87}.
For every fixed $k$, there is a linear time algorithm deciding whether the tree-width is at most $k$, and if that is the case, producing a corresponding tree decomposition \cite{Bodlaender96}. For arbitrary $k$, this task can still be approximated. A tree decomposition of width $O(k \log n)$ can be found in polynomial time 
\cite{BodlaenderGHK95}, and  in time $O(c^k n)$ almost a $5$-approximation \cite{BodlaenderDDFLP2013} can be found (a tree of width at  most $5k+4$ to be precise).

It is often convenient to view $T$ as a rooted tree, where an arbitrary fixed node has been chosen as the root.
\begin{definition}
 A \emph{semi-smooth tree decomposition} of width $k$ is a rooted tree decomposition where the bag $B_i$ of every node $i$ contains exactly $1$ vertex that is not in the bag of the parent node.
 For rooted trees $T$ with $v \in B_i \setminus B_{p(i)}$ for $p(i)$ being the parent of $i$, we say that \emph{node $i$ is the \emph{home} of vertex $v$}. 
 \end{definition}
 In other words, the home of a vertex $v$ is the highest node whose bag contains $v$. 
 The bag $B_r$ of the root $r$ of a semi-smooth tree decomposition contains just one vertex.

\begin{proposition}
 Every graph $G=(V,E)$ has a semi-smooth tree decomposition of width $k=\tw(G)$ with $|I|=|V|-k$. Any tree decomposition can be transformed into a semi-smooth tree decomposition in linear time. 
 \end{proposition}

\begin{proof}
 Do a depth-first search of the tree and omit nodes whose bag is contained in the bag of the parent. Insert intermediate nodes if more than one vertex has the same home. 
 \end{proof}

We use the standard notation of $k$-expression to define clique-width.
\begin{definition}
A \emph{$k$-expression} is an expression formed from the atoms  $i(v)$, the two unary operations $\et{i}{j}$ and $\rh{i}{j}$, and one binary operation $\oplus$ as follows.
\begin{itemize}
\item $i(v)$ creates a vertex $v$ with label $i$, where $i$ is  from the set $\{1, \dots , k\}$.
\item $\et{i}{j}$ creates an edge between every vertex with label $i$ and every vertex with label $j$ for $i \neq j$ (with $i,j \in \{1, \dots , k\}$).
\item $\rh{i}{j}$ changes all labels $i$ to $j$ (with $i,j \in \{1, \dots , k\}$).
\item $\oplus$ (join-operation) does a disjoint union of the generated labeled graphs. 
\end{itemize}
Finally, the \emph{generated graph} is obtained by deleting the labels.
\end{definition}
We also allow multi-way join-operations, as $\oplus$ is associative.

\begin{definition}
The \emph{clique-width} $\cw(G)$ of a graph is the smallest $k$ such that the graph can be defined by a $k$-expression \cite{CourcelleO2000}.
\end{definition}

Computing the clique-width is NP-hard \cite{FellowsRRS2006}.
Thus, one usually assumes that a graph is given together with a $k$-expression.

Theoretically, this is not necessary, because
for constant $k$, the clique-width can be approximated by a constant factor in polynomial time \cite{OumS2006,Oum2008}. But these factors are exponential in $k$.

The new notion of multi-clique-width is defined similarly to the clique-width. The essential difference is that every vertex can have any set of labels (including singleton sets and the empty set). There is a new operation $\ep{i}$ to delete a label.
The creation of multiple vertices with the same labels by one command is an unessential convenience.

\begin{definition}
A \emph{multi-$k$-expression} is an expression formed from the atoms  $\im{m}{i_1,\dots,i_{\ell}}$, the three unary operations $\et{i}{j}$, $\rh{i}{S}$, and $\ep{i}$, as well as the binary operation $\oplus$ as follows. Assume $i \in \{1, \dots , k\}$, $j \in \{0, \dots , k\}$ and $\emptyset \subseteq S, \{i_1,\dots,i_{\ell}\} \subseteq \{1, \dots , k\}$). 
\begin{itemize}
\item $\im{m}{i_1,\dots,i_j}$ with $m$ a positive integer and $i_1 < \dots < i_j \leq k$, creates $m$ vertices, each with label set $\{i_1,\dots,i_j\}$. 
\item $\et{i}{j}$ creates an edge between every vertex $u$ with label $i$ and every vertex $v$ with label $j$.
This operation is only allowed when there are no vertices with label $i$ and $j$ simultaneously, in particular $i \neq j$. 
\item $\rh{i}{S}$ replaces replaces label $i$ by the set of labels $S$, i.e., if a vertex $v$ had label set $S'$ with $i \in S'$ before this operation, then $v$ has label set $(S' \setminus {i}) \cup S$ after the operation. 
\item $\ep{i}$ deletes the label $i$ from all vertices.
\item $\oplus$ (join-operation) does a disjoint union of the generated labeled graphs.
\end{itemize}
Finally, the \emph{generated graph} is obtained by deleting the labels.
\end{definition}

$S$ and $\{i_1,\dots,i_{\ell}\}$ are allowed to be empty, even though the latter is not very interesting, as it only creates isolated vertices.
Note that $\ep{i}$ is just the special case of $\rh{i}{S}$ with $S=\emptyset$. We list it separately, because one might want to consider the  \emph{strict multi-$k$-expressions}  without $\rh{i}{S}$. In Theorem~\ref{thm:mtw} below, $\rh{i}{S}$ is not used. Alternatively, one might restrict $\rh{i}{S}$ to the classical case with $S$ being a singleton. The relative power of these 3 versions might be worth studying.

\begin{definition}
The \emph{multi-clique-width} $\mcw(G)$ of a graph is the smallest $k$ such that the graph can be defined by a multi-$k$-expression.
\end{definition}

We also define boolean-width in order to compare it with multi-clique-with.
\begin{definition}
 A \emph{decomposition tree} of a graph $G=(V,E)$ is a tree $T$ where $V$ is the set of leaves and where all internal nodes have degree 3. \\
 Every edge $e$ of $T$ defines a partition of $V$ in to $X$ and $\overline{X}$ consisting of the leaves of the two trees obtained from $T$ by removing $e$. \\
 The \emph{set of unions of neighborhoods} of $X$ across the cut $\{X, \overline{X}\}$ is the set 
 \[ U(X) = \{ S' \subseteq \overline{X} \: |\: \exists S \subseteq X \;\; S' = N(S) \cap \overline{X}\}. \]
 $\mbox{bool-dim}(X) = \log_2 |U(X)|$. \\
 The \emph{boolean-width} of $G$ is the minimum over all trees $T$ of the maximum over all cuts $\{X, \overline{X}\}$ defined by an edge $e$ of $T$ of $\mbox{bool-dim}(X) = \log_2 |U(X)|$.
\end{definition}

\section{Relationship between Different Width Parameters}
Multi-clique-width extends the notions of tree-width and of clique-width in a natural way.

\begin{proposition} \label{prop:mcw-cw}
 For every graph $G$, $\mcw(G) \leq \cw(G)$.
\end{proposition}

\begin{proof}
This follows immediately from the definition. 
\end{proof}

\begin{proposition}
 For every graph $G$, $\cw(G) \leq 2^{\mcw(G)}$.
\end{proposition}

\begin{proof}
Use a new label for every set of labels. 
\end{proof}

\begin{corollary} \label{cor:bounded}
A class of graphs has bounded clique-width if and only if it has bounded multi-clique-width.
\end{corollary}

\begin{proof}
This is implied by the previous two propositions.  
\end{proof}

\begin{corollary}
 Properties of graphs expressible in monadic second order logic without quantifiers over sets of edges are linear time decidable for graphs of bounded multi-clique-width.
\end{corollary}

\begin{proof}
This follows from Corollary~\ref{cor:bounded} and the corresponding meta-theorem for clique-width \cite{CourcelleMR2000}. 
\end{proof}

\begin{theorem} \label{thm:mtw}
If tree-decomposition of width $k$ of a graph $G=(V,E)$ is given, then a multi-$(k+2)$-expression  for $G$ can be found  in polynomial time.
\end{theorem}

\begin{proof}
 Assume, $G$ is given with a tree decomposition of width $k=\tw(G)$. In linear time, the tree decomposition is transformed into a semi-smooth tree decomposition  $(\{B_{i } \, : \, i \in I \}, T)$. Now we assign an identifier $\iota(v)$ from $\{1,2,\dots, k+1\}$ to each vertex $v$ top-down, i.e., starting at the root of $T$. 
 When identifiers have been assigned to the vertices whose home is above vertex $v$, we assign to vertex $v$ the smallest identifier not assigned to the other vertices in the bag of the home of $v$.
 
 Next, we define a multi-$(k+2)$-expression whose parse tree $T'$ is basically isomorphic to the tree $T$ of the tree decomposition. The difference it that in $T'$ every internal node has an additional child that is a leaf. We call it an auxiliary leaf. Furthermore, above each internal node $i$, we introduce three auxiliary nodes obtained by subdividing the edge to the parent of~$i$. 

 The main idea is that every vertex $v$ is created at its home, or more precisely, in the auxiliary node below its home. Then the edges from $v$ to neighbors of $v$ with a home further down the tree 
 are added.  The upper neighbors of $v$, i.e., those that have their home higher up the tree, are not yet created. Vertex $v$ remembers to attach to these neighbors later by taking the set of identifiers of these neighbors as its labels. All upper neighbors of $v$ are together with $v$ in the bag $B_{h(v)}$ of the home $h(v)$ of $v$ in $T$. The vertex $v$ needs at most $k$ labels for this purpose. We give $v$ an additional label, $k+2$, to allow the lower neighbors of $v$ to connect to $v$.
 Node $i$ of $T'$ is a multi-way join operation of all its children, including the new auxiliary child.
 The purpose of the three nodes inserted above node $i$ is to add the edges between $v$ and its neighbors in the subtree of $i$, and to delete the two labels that have been used to create these new edges.
The multi-$(k+2)$-expression is built bottom-up.
 
 Now we define the multi-$(k+2)$-expression exactly by assigning atoms to the leaves and operations to the internal nodes as follows.
 \begin{description}
\item[Regular leaf:] 
Let the leaf $i$ be the home of some vertex $v$. Let $v_1, \dots, v_{\ell}$ be the neighbors of $v$ with identifiers $i_1,\dots,i_{\ell}$. Clearly, $\{v, v_1, \dots, v_{\ell}\} \subseteq B_i$.
Then the expression $\im{1}{i_1,\dots,i_{\ell}}$ creates $v$ in leaf $i$.

\item[Auxiliary leaf:]
Let the internal node $i$ be the home of some vertex $v$. Let $v_1, \dots, v_{\ell}$ be the upper neighbors of $v$ with identifiers $i_1,\dots,i_{\ell}$. Let $c_0(i)$ be the child of $i$ which is an auxiliary leaf. Then the expression for $c_0(i)$ is $\im{1}{k+2,i_1,\dots,i_{\ell}}$.

\item[Internal node:]
Let $i$ be the home of some vertex $v$, and let $c_1(i),\dots,c_{q}(i)$ be the children of $i$ in $T$. Let $c_0(i)$ be the auxiliary leaf child of $i$ in $T'$. Furthermore, let $\iota(v)$ be the identifier of $v$. Assume, for child $c_j(i)$ we already have the expression $E_j$.
Then the multi-$(k+2)$-expression for node $i$, or more precisely of the third auxiliary node above it, is 
\[\epsilon_{k+2}(\epsilon_{\iota(v)}(\eta_{\iota(v),k+2}(E_0 \oplus E_1 \oplus \dots \oplus E_q))).\]

\end{description}
 Now the following is easily proved by induction on the height of node $i$.
\begin{claim}
 The multi-$(k+2)$-expression for node $i$ generates the labeled graph $G_i=(V_i,E_i)$ induced by the vertices whose home is in the subtree of $i$. Furthermore, the set of labels of every vertex $v \in V_i$ is equal to the set of identifiers of the neighbors of $v$ in $V \setminus V_i$.
\end{claim}
By the inductive hypothesis of the claim, all vertices $V'$ in the subtree of node $i$ that are adjacent to $v$ in $G$ have a label $\iota(v)$. The vertex $v$ has a label  $k+2$, but no label $\iota(v)$. Thus the operation $\eta_{\iota(v),k+2}$ creates exactly the edges between $v$ and $V'$. Now, the labels $\iota(v)$ and $k+2$ can be deleted, because both have served their purpose. From every vertex labeled $\iota(v)$, the edge to $v$ is now already constructed, and the label $k+2$ only had to mark the vertex $v$ for the construction of these edges.

The claim for the root implies the theorem.
\end{proof}

A weaker form of Theorem~\ref{thm:mtw} is the implied inequality between multi-clique-width and tree-width.
\begin{corollary}
 \label{cor:mtw}
 For every graph $G$, $\mcw(G) \leq \tw(G)+2$. 
\end{corollary} \qed

As an immediate corollary, we obtain $cw \leq 2^{\tw(G)+2}$. The tighter bound of $cw \leq 2^{\tw(G)+1}+1$ \cite{CourcelleO2000} is obtained by noticing that one could use the label $k+2$ strictly as a singleton label. Instead of deleting it with an $\epsilon_{k+2}$ operation, one could change it to the set of other labels we wanted to assign to that vertex using a $\rh{i}{S}$ operation. The even tighter bound $cw \leq 1.5 \cdot 2^{\tw(G)}$ \cite{CorneilR2005} is obtained by handling higher degree join nodes more efficiently. Following every binary join, the necessary edges could be inserted, allowing the number of labels to be decreased. This saves one fourth of the labels. 

\begin{corollary} \label{thm:lower}
 There are graphs $G$ with $\cw(G) \geq 2^{\lfloor \mcw(G)/2 \rfloor - 2}$.
\end{corollary}

\begin{proof}
There are graphs $G$ with $\tw(G)=k$ and clique-width $\cw(G) \geq 2^{\lfloor k/2 \rfloor - 1}$ \cite{CorneilR2005}. 
Such graphs have multi-clique-width $\mcw(G) \leq k+2$ by Corollary~\ref{cor:mtw}.
\end{proof}

Naturally, it is easy to find graph classes with unbounded tree-width that still exhibit this exponential discrepancy between clique-width and multi-clique-width. One way is just to add a large clique, but there are many not so obvious ways.

\begin{corollary}
 There are graphs $G$ with $\tw(G)=k=\Omega(n)$ and clique-width $\cw(G) \geq 2^{\lfloor k/2 \rfloor - 1}$. \qed
\end{corollary}

We want to compare multi-clique-width with boolean-width.
\begin{theorem}
 For every graph $G$, $\boolw(G) \leq \mcw(G) \leq 2^{\boolw(G)}$.
\end{theorem}

\begin{proof}
 $\boolw(G) \leq \mcw(G)$:  Assuming $\mcw(G)=k$, we start with a multi-$k$-expression for $G$. W.l.o.g., assume that each vertex $v$ is created as a single vertex with the operation $i(v)$. Then, there is a bijection between the vertices $V$ and the leaves of the parse tree $T$. Viewed as a graph, the other nodes of $T$ have degrees 2 or 3. We replace all maximal paths with internal nodes of degree 2 by single edges to obtain a tree $T'$. 
 
 Consider any edge $e=(u,v)$ of $T'$, where $u$ is a descendant of $v$ in $T$. Let $X \subseteq V$ be the set of vertices of the subtree $T_v$. For every subset $S \subseteq X$, the set $N(S) \cap \overline{X}$ of neighbors of $S$ outside of $X$ only depends on the union of the set of labels of the vertices of $S$. There are at most $2^k$ such subsets of labels, and thus at most $2^k$ such neighborhoods. The binary logarithm of the largest such number of neighborhoods over all edges of $T'$ is an upper bound on $\boolw(G)$, i.e., $\boolw(G) \leq k$.

$\mcw(G) \leq 2^{\boolw(G)}$:  By Lemma~\ref{prop:mcw-cw} $\mcw(G) \leq \cw(G)$, and the inequality
$\cw(G) \leq 2^{\boolw(G)}$ \cite{Bui-XuanTV2011} is known.
\end{proof}

Even though, boolean-width has the desirable property $\boolw(G) \leq \mcw(G)$, sometimes more efficient algorithms are possible in terms of $\mcw(G)$ than in terms of $\boolw(G)$. 
 Indeed, every graph property expressible in $\mathcal MS_1$, is decidable in linear time for graphs of bounded clique-width \cite{CourcelleMR2000}, while for arbitray graphs of bounded boolean-width, at least quadratic time is required even to read the input. Naturally, this can also be viewed as an indication of the strength of the boolean-width parameter. Even graphs without a simple structure can have small boolean-width. In the next section, we will see that for specific problems the (exponential) dependance of the running time on the multi-clique-width can be very good.

\section{Algorithms based on Multi-Clique-Width}

The algorithmic purpose of clique-width and other width parameters is to put problems into FPT, i.e., making them fixed parameter tractable (see \cite{DowneyF99}). This means achieving a running time of $O(f(k)n^{O(1)})$ for an arbitrary function $f$. In reality things are not so bad.
Algorithms based on clique-width often have a running time of $O(c^k n^e)$ or $O(c^{k \log k} n^e)$ with $k=\cw(G)$, $n=|V|$, and $c$ and $e$ being small constants. Assume that we are given a multi-$k$-expression for $G$ and we have an algorithm with similar running time when $k$ is the multi-clique-width. Then we have an exponential time speed-up when choosing the multi-clique based algorithm with running time $2^{O(k)} n^e$,  
instead of the clique-width based algorithm with clique-width $2^{\Omega(k)}$ and running time $2^{2^{\Omega(k)}} n^e$  
for infinitely many graphs.

Indeed, we want to illustrate here that this scenario is occurring quite naturally. 
We exhibit it for Independent Set. The running time as a function of the width is roughly the same for clique-width $k$ as for multi-clique-width $k$. Hence, we gain an exponential speed-up in the width parameter for all the many instances were the clique-width is exponentially bigger than the multi-clique-width.

Instead of only finding a maximum independent set, or even just computing its size, we solve the more involved problem of computing the independent set polynomial, i.e., computing the numbers of independent sets of all sizes. This is not much more difficult, and one can easily simplify the algorithm if only a maximum independent set is needed. Then the dependence of the running time on the size $n$ goes down to linear from polynomial, while the dependence on the width $k$ stays singly exponential. In particular, we have an FPT algorithm to compute the independent set polynomial.
We refer to \cite{CourcelleMR01,FischerMR08} for more discussions of the fixed parameter tractability of counting problems.

\begin{definition}
The independent set polynomial of a graph $G$ is 
\[I(x) = \sum_{i=1}^n  a_i x^i \]
where $a_i$ is the number of independent sets of size $i$ in $G$.
\end{definition}
The independent set polynomial is not strong enough to describe the involvement of the different labels in the independent sets. We need to do a more detailed counting to allow recurrence equations to govern the definition of the polynomials as the labeled graph is assembled by a multi-$k$-expression.

\begin{definition}
Let $[k] = \{1,\dots , k\}$ be the set of vertex labels. The $[k]$-labeled independent set polynomial of a $[k]$-labeled graph $G$ (each vertex can have multiple labels from $[k]$) is
\[P(x,x_1, \dots , x_k) = \sum_{i=1}^n \sum_{(n_1,\dots,n_k) \in \{0,1\}^k} a_{i;n_1,\dots,n_k} 
	\; x^i \prod_{j=1}^k x_j^{n_j} \]
where $n_j \in \{0,1\}$ and $a_{i;n_1,\dots,n_k}$ is the number of independent sets of size $i$ in $G$ which contain some vertices with label $j$ if and only if $n_j = 1$.
\end{definition}
Now, the independent set polynomial $I(x)$ can be expressed immediately by the $[k]$-labeled independent set polynomial $P(x,x_1, \dots , x_k)$.

\begin{proposition} \cite{Furer2014}
 The independent set polynomial $I(x)$ of a $[k]$-labeled graph $G$ is
 $I(x) =  P(x,1,\dots,1)$.
 \end{proposition}

\begin{proof}
 $I(x) =  \sum_{i=1}^n \sum_{(n_1,\dots,n_k) \in \{0,1\}^k} a_{i,n_1,\dots,n_k} x^i = P(x,1,\dots,1)$, because \\$a_i = \sum_{(n_1,\dots,n_k) \in \{0,1\}^k} a_{i,n_1,\dots,n_k}$. 
\end{proof}

\begin{theorem} \label{thm:indset}
 Given a graph $G$ with $n$ vertices, and  a multi-$k$-expression generating $G$ with multi-clique-width $k$, the independent set polynomial $I(x)$ of $G$ can be computed in time $O(2^k (kn)^{O(1)})$.
 \end{theorem}

\begin{proof}
Using dynamic programming, we compute the $[k]$-labeled independent set polynomial of the $[k]$-labeled graphs generated by subexpressions of the given multi-$k$-expression. The computation is done bottom-up in the parse tree of the given multi-$k$-expression.

For any atomic expression $\im{m}{i_1,\dots,i_j}$ creating $m$ vertices with labels $i_1,\dots,i_j$, we have the $[k]$-labeled independent set polynomial 
 \[ 1 + \sum_{\ell=1}^m \binom{m}{\ell} x^\ell x_{i_1} \cdots x_{i_j} = 1 + ((1+x)^m - 1)x_{i_1} \cdots x_{i_j} . \]
In $O(m)$ arithmetic operations, we can compute all coefficients using the recurrence $\binom{m}{\ell+1} = \binom{m}{\ell} (m-\ell)/(\ell+1)$.
Thus all atomic expressions for the $n=|V|$ vertices can be computed in time $O(n)$.

If the expression $E$ has the polynomial $\tilde{P}(x,x_1,\dots,x_k)$, then $\et{i}{j}(E)$ has the polynomial 
\[P(x,x_1,\dots,x_k) =
\tilde{P}(x,x_1,\dots,x_k) \mod x_i x_j,\] 
i.e., terms containing $x_i$ and $x_j$ are deleted. This is correct, because a set of vertices is independent after the introduction of the edges between labels $i$ and $j$, if and only if it was an independent set before and does not contain both labels $i$ and $j$.

If the expression $E$ has the polynomial $\tilde{P}(x,x_1,\dots,x_k)$, then $\rh{i}{S}(E)$ has the polynomial 
\begin{align} \label{eqn:square}
P(x,x_1,\dots,x_k) & =  \tilde{P}(x,x_1,\dots x_{i-1},  x_{i_1} \cdots x_{i_j}, x_{i+1},\dots,x_k)  \nonumber \\
	&  \; \mod (x_{i_1}^2 - x_{i_1}) \, \cdots \!\!\!\!\! \mod (x_{i_j}^2 - x_{i_j}) 
\end{align}
for $S=\{i_1, \dots, i_j\}$, i.e., first $x_i$ is replaced by the product $x_{i_1} \dots x_{i_j}$. Then squares of indeterminates are replaced by their first powers. This is correct, because we still count all independent sets. They just occur in different categories as they involve different labels.

If the expression $E$ has the polynomial $\tilde{P}(x,x_1,\dots,x_k)$, then $\ep{i}(E)$ has the polynomial 
\[P(x,x_1,\dots,x_k) =
\tilde{P}(x,x_1,\dots x_{i-1}, 1, x_{i+1},\dots,x_k),\] 
i.e., the indeterminate $x_i$ is replaced by 1. This is correct, because it is just a special case $\rh{i}{S}(E)$.

If the expression $E_{\ell}$ ($\ell \in \{1,2\}$) has the polynomial $\tilde{P}_{\ell}(x,x_1,\dots,x_k)$, then the expression $E_1 \oplus E_2$ has the polynomial 
\begin{align} \label{eqn:product}
P(x,x_1,\dots,x_k) &=
\tilde{P}_{1}(x,x_1,\dots,x_k) \, \tilde{P}_{2}(x,x_1,\dots,x_k) \nonumber\\
	& \; \mod (x_{1}^2 - x_{1}) \, \cdots \!\!\!\!\! \mod (x_{k}^2 - x_{k}), 
\end{align}
i.e., in the product of the polynomials, every $x_i^2$ is replaced by $x_i$, as we only care about the occurrence of a label and not about the multiplicity of such an occurrence. This is correct, because every  independent set of $G_1$ can be combined with every independent set of $G_2$ to form an independent set of the join graph $G$, and every independent set of $G$ can be formed in this way. 

To bound the running time, one should notice that the polynomial $P(x,x_1,\dots,x_k)$ has $2^k (n+1)$ coefficients. 
The polynomial $\tilde{P}(x,x_1,\dots x_{i-1},  x_{i_1} \cdots x_{i_j}, x_{i+1},\dots,x_k)$ in Eq.~(\ref{eqn:square}) has only $2^{k+1} (n+1)$ coefficients, not $3^k (n+1)$, as only a few monomials which are quadratic in some $x_i$'s appear. If the product in Eq.~(\ref{eqn:product}) is computed by school multiplication, then the running time is $O(3^k (kn)^{O(1)})$. But with a fast Fourier transform (evaluating the polynomial for $x_i=0$ and $x_1=1$ for all $i$), the time is only $O(2^k (kn)^{O(1)})$.
\end{proof}

The easier problem of just finding the size of a maximum independent set (rather than computing the numbers of independent sets of all sizes) is now trivial. At each stage, for all exponents $n_1, \dots, n_k$, the coefficient $a_{i;n_1, \dots n_k}$ is only stored for the largest $i$ with $a_{i;n_1, \dots, n_k} \neq 0$.

\begin{corollary}
A maximum independent set can be found in time $O(2^k k^{O(1)} n)$ in graphs with multi-clique-width $k$.
\end{corollary}

\begin{proof}
 If during the dynamic programming algorithm to compute the size of a maximum independent set, one always stores where the larger exponent $i$ came from, then at the end, one can easily backtrack to actually find a maximum independent set.
\end{proof}

As an additional example, we consider the NP-complete decision problem $c$-coloring, asking whether the input graph $G$ can be colored with $c$ colors for a constant integer $c \geq 3$, such that no adjacent vertices have the same color.

\begin{theorem}
 For graphs $G$ of multi-clique-width $k$ with a given multi-$k$-expression for $G$, and any positive integer constant $c$, the $c$-coloring problem can be solved in time $2^{O(ck)}n$.
\end{theorem}

\begin{proof}
 We present a dynamic programming algorithm based on the parse tree structure of the multi-$k$-expression. We classify the colorings of the graphs generated by sub-expressions according to the labels used for the vertices of each color.
 Let $Q$ with $|Q|=c$ be the set of colors and $L$ with $|L|=k$ be the set of labels.
 Let $B_1, \dots, B_r$ with $r=2^{ck}$ be the sequence (say in lexicographic order) of all bipartite graphs with the left vertex set $C$ and the right vertex set $L$. Let $E_p$ be the set of edges in $B_p$.
 For every subexpression $F$, we define $F(B_p)$ so that it is true, if and only if the graph generated by $F$ can be colored with $Q$ such that some vertex colored with $q \in Q$ is labeled with a set of labels containing $i \in L$, if and only if $(q,i)$ is an edge in $B_p$. 
 
 We now show that $F(B_p)$ can easily be computed from all the $F'(B_{p'})$ where $F'$ is a subexpression of $F$ and $j \in \{1,\dots,r\}$. We analyze according to the structure of $F$.
 
 If $F$ is an atomic expression $\im{m}{i_1,\dots,i_j}$ creating $m$ vertices with labels $i_1,\dots,i_j$, then
 $F(B_p)$ is true, if and only if $E_j = \{(q,i) \: | \: i \in \{i_1,\dots,i_j\} \}$ for some $q \in Q$.
 
 If $F= \et{i}{j}(F')$, then $F(B_p)$ is true, if and only if $F'(B_p)$ is true and for no color $q \in Q$ there are both edges $(q,i)$ and $(q,j)$ present in $B_p$. 
 In other words, a previous coloring is still valid, if and only if no color appears at both endpoints of newly added edges.
 
 If $F = \rh{i}{j}(F')$, then $F(B_p)$ is true, if and only if $F'(B_{p'})$ is true for some $p'$ with
 \[E_p = \{(q,\ell) \: |\: \mbox{$(q,\ell) \in E_{p'}$ and $\ell \notin \{i,j\}$} \} 
 \cup  \{(q,j) \: |\: \mbox{$(q,i) \in E_{p'}$ or $(q,j) \in E_{p'}$} \}
  \]
 
 If $F = F' \oplus F''$, then $F(B_p)$ is true, if and only if $F'(B_{p'})$ is true and $F''(B_{p''})$ is true
 for some $p', p''$ with $E_p = E_{p'} \cup E_{p''}$.
 
 Given this simple characterization of $F(B_p)$ in terms of $F'(B_{p'})$ for some $p'$ and the immediate sub-expressions $F'$, it should be immediately clear how the value of $F(B_p)$ can be computed, when the values of the $F'(B_{p'})$ are known.
 
 Furthermore, it is a simple proof by induction on the structure of an expression $F$ that $F(B_p)$ is true, if and only if 
the graph generated by $F$ can be colored with $Q$ such that some vertex colored with $q \in Q$ is labeled with a set of labels containing $i \in L$, if and only if $(q,i)$ is an edge in $B_p$. 

Naturally, at the end, the graph generated by $F$ is $k$-colorable, if and only if $F(B_p)$ is true for some $B_p$.

The running time is linear in $n$, because there are $O(n)$ nodes to process and the time spent in every node only depends on the number $c$ of colors and the number $k$ of labels. In every node, an array of $2^{ck}$ boolean values (one for each bipartite graph on the vertex sets $Q$ and $L$) has to be processed in a simple fashion. The resulting running time is $2^{O(ck)} n$.

There is quite some waste of time involve in handling all the bipartite graph on the vertex sets $Q$ and $L$, because the truth value for a graph $B_j$ does not change, when the set of colors $Q$ and the set of labels $L$ are permuted in an arbitrary way. This does not mean that the running time can be divided by $c! k!$, because typically many such permutations are automorphisms not creating new bipartite graphs. The exact number of isomorphism types of such bipartite graphs can be computed with the Redfield-P\'{o}lya enumeration theorem (see \cite{PoRe87}), but that does not result in a nicer upper bound. Clearly, any practical implementation would do the computation for just one bipartite graph for every isomorphism type.
 \end{proof}

\section{Conclusions and Open Problems}
We have proposed a powerful parameter multi-clique-width. It allows us to achieve faster running times for natural classes of graphs and interesting algorithmic tasks. Assume, we are given the input graph by a multi-$k$-expression. Then we have very efficient algorithms for this class of graphs, as illustrated by the independent set polynomial and the coloring problem. On the other hand, for any algorithm based on clique-width, we could only get exponentially slower (in $k$) algorithms for the same problems and the same collection of graphs. Also, equally efficient algorithms are not known based on rank-width or boolean-width, when the corresponding decompositions are given.

Most questions related to the new multi-clique-width are still open. Is it difficult to compute or approximate? We expect it to be NP-hard, like clique-width. We also conjecture it to be in FPT (fixed parameter tractable) and to be constant factor approximable in time singly exponential in the multi-clique-width and linear in the length like tree-width. But obviously this is very difficult, as it is also open for clique-width.

A main question is whether most algorithms for clique-width $k$, can be extended to work with similar efficiency for multi-clique-width $k$. We have illustrated that this is the case for some interesting counting and decision problems. On the other hand, there is the question of identifying the problems where this is not the case.

\bibliography{mcw}
\end{document}